\documentclass[11pt]{llncs}
\usepackage{amsmath}
\usepackage{amssymb}
\usepackage{bm}
\usepackage{graphicx}
\usepackage{epstopdf}
\usepackage{rotating}
\usepackage{cite}
\usepackage{color}
\usepackage{fancybox}
\usepackage{pstricks}
\usepackage{pst-plot}

\usepackage{pstricks-add}
\usepackage{epsf}
\usepackage{flexisym}
\usepackage[linesnumberedhidden,titlenotnumbered,ruled]{algorithm2e}
\SetKw{Continue}{continue}
\SetKw{Break}{break}
\SetKw{DownTo}{downto}
\usepackage{hyperref}

\newcommand{\unif}[1]{ \overset{#1}{\leftarrow}}
\newcommand{\Z}[2]{\mathbb{Z}_{#1}^{#2}}

\begin{document}
%
\title{Attacks on a Privacy-Preserving Publish-Subscribe System and\\ a Ride-Hailing Service}

%



%

\author{Srinivas Vivek}
\institute{IIIT Bangalore, IN \\\small\texttt{srinivas.vivek@iiitb.ac.in} }


\pagestyle{plain}
\maketitle              
%

\begin{abstract}

A privacy-preserving Context-Aware Publish-Subscribe System (CA-PSS) enables an intermediary (broker) to match the content from a publisher and the subscription by a subscriber based on the current context while preserving confidentiality of the subscriptions and notifications. While a privacy-preserving Ride-Hailing Service (RHS) enables an intermediary (service provider) to match a ride request with a taxi driver in a privacy-friendly manner. In this work, we attack a privacy-preserving CA-PSS proposed by Nabeel et al. (2013), where we show that any entity in the system including the broker can learn the confidential subscriptions of the subscribers. We also attack a privacy-preserving RHS called lpRide proposed by Yu et al. (2019), where we show that any rider/driver can efficiently recover the secret keys of all other riders and drivers. Also, we show that any rider/driver will be able to learn the location of any rider. The attacks are based on our cryptanalysis of the modified Paillier cryptosystem proposed by Nabeel et al. that forms a building block for both the above protocols.
\end{abstract}

%
%
%

\section{Introduction}

Publish-Subscribe systems (also written as pub/sub systems) are a well-known paradigm 
to disseminate information among multiple parties in a distributed and asynchronous manner.
Subscribers subscribe for content from publishers, who create content and push notifications to the 
intermediaries' network (aka. brokers). Brokers route the content to subscribers based on their 
subscription. A Context-Aware Publish-Subscribe System (CA-PSS) extends a pub/sub system by 
taking into account the subscriber context. For instance, the context could be the location of a subscriber
in a traffic information service, and the content from a publisher could correspond to the 
traffic situation in the neighbourhood of the subscriber. Hence, the context of a subscriber 
could change frequently over time. It is important to protect the confidentially of the context, 
subscription and notification. The goal of a Privacy-Preserving CA-PSS (PP-CA-PSS) is to protect the above
confidential information. The scheme from Nabeel et al. \cite{NabeelABB13} is one of the early 
proposals of a privacy-preserving CA-PSS. 
There have many works on this topic and we refer
to \cite{munster2018,cui2019} and references therein for more details.

On the other hand, a Ride-Hailing Service (RHS) too provides location-based services. RHSs have become increasingly popular in recent years. Uber, Ola, 
Lyft, Didi, Grab, etc. are some popular RHSs. 
With these services also comes the
risk of misuse of personal data. The Ride matching Service Providers (RSPs) collect personal information
regarding the riders and drivers along with their ride statistics. There have been many instances 
of violation of individual privacy of the users using these RHSs \cite{norton_databreach}.
A Privacy-Preserving RHS (PP-RHS) aims to provide privacy guarantees to the users of the RHS, namely, 
riders and drivers. Recent years have witnessed many proposals of PP-RHS \cite{Pham2017PrivateRideAP,ORidePaper,khazbak,raza_bride,he_ridesharing,pRideLuo,lpRideYu,wangTrace,yuPSRIde}.
Recently, the ORide RHS's \cite{ORidePaper} security was revisited in \cite{DeepakMV19}.
  
The focus of this work is on the PP-CA-PSS from \cite{NabeelABB13} and the
 lpRride PP-RHS proposed by Yu {et} al. \cite{lpRideYu}. These two seemingly disparate protocols share
 the following common features: a) both offer or potentially can offer location-based services, and 
 b) both the protocols are based on the modified Paillier cryptosystem that was proposed in 
 \cite{NabeelABB13}. The modified Paillier cryptosystem 
 (yet another variant of the Paillier cryptosystem \cite{Paillier99}) is an additively 
homomorphic (digital signature-like) cryptosystem and both the protocols 
use this cryptosystem to blind the 
subscriptions and notifications (resp. locations of riders and taxis) but 
still can perform subscription-notification matching (resp. ride-matching) on the blinded data using the 
additive homomorphism property. 
 

\subsection{Our Contribution}
In this work, we analyse the security of the PP-CA-PSS from \cite{NabeelABB13}  
and the lpRide protocol. For the PP-CA-PSS we demonstrate an attack where any 
entity in the system including the broker can fully learn the confidential subscriptions of the 
subscribers. This invalidates
the claim in \cite{NabeelABB13} that subscriptions remain confidential. We would like to note
that we do not target notifications from publishers or the content. Yet, subscriptions
can leak confidential information such as locations of subscribers.

On the lpRide protocol we exhibit a key recovery attack 
by any rider or driver that can efficiently recover the secret keys of all other riders and drivers. 
Also, we show that any rider or driver 
will be able to learn the location of any rider. We were unable to recover drivers' 
locations as they are blinded by random values.

All our adversaries are honest-but-curious. The basis of our attack is our cryptanalysis of the modified 
Paillier cryptosystem mentioned above. In particular, we show that anyone will be able to forge the ``signatures'' and that these signatures are deterministic.
Hence, this result
is of independent interest to the security of protocols that are based on the modified Paillier 
cryptosystem. It is somewhat surprising that a simple attack on this cryptosystem went unnoticed 
despite many follow-up works of \cite{NabeelABB13}.

In Section \ref{sec:modified}, we recall and then cryptanalyse the modified Paillier cryptosystem. In 
Section \ref{sec:capss}, we briefly recall the PP-CA-PSS, and then describe our attack. 
In Section \ref{sec:lpride}, we briefly recall the lpRide protocol, and then describe our attack.
Section \ref{sec:con} concludes the paper.
\section{Cryptanalysis of Modified Paillier Cryptosystem}
\label{sec:modified}

\subsection{Recall of the Modified Paillier Cryptosystem}
As mentioned earlier, 
the modified Paillier  cryptosystem 
was proposed in \cite{NabeelABB13}. A preliminary use of this cryptosystem was already made in 
\cite{NabeelSB12}. As the name suggests, this cryptosystem is a variant of the Paillier encryption scheme 
\cite{Paillier99}. It consists of the following three algorithms:
\begin{itemize}
\item \textbf{Key generation}: choose two distinct large primes $p$ and $q$. Compute $N = p \cdot q$ and 
$\lambda = \mathrm{lcm}(p-1,q-1),$ the Carmichael function of $N$. Randomly sample a base 
$g_p \unif{\$} \Z{N^2}{*}$, such that the order of $g_p$ is a multiple of $N$. The latter condition can be 
ensured by checking the condition
\[
\gcd(L(g_p^\lambda \pmod{N^2}), \; N) = 1,
\]   
where 
\begin{equation}
\label{eq:l}
L(x) = \frac{x-1}{N},
\end{equation}
for
\[
x \in \{y\;|\;y \in \Z{\ge 0}{},\; y < N^2,\; y \equiv 1 \pmod{N}\}.
\]
Compute
\begin{equation}
\label{eq:mu}
\mu = \left( L(g_p^\lambda \pmod{N^2}) \right)^{-1} \pmod{N}.
\end{equation}
The public ``decryption'' key is
\[
PK'=(N,g_p,\mu),
\]
and the secret ``encryption'' key is 
\[
SK' = \lambda.
\]
\item \textbf{Encryption} $E'(m,r,SK')$: let the plaintext $m \in \Z{N}{}$. Sample a random value 
$r \unif{\$} \Z{N}{*}$. Compute the ``ciphertext'' 
\begin{equation}
\label{eq:enc}
c = g_p^{m\lambda}r^{N\lambda} \pmod{N^2}.          	 
\end{equation}
When the randomness and the secret key is implicit from the context, we simply denote the encryption
of a message $m$ as $E'(m)$. 

\item \textbf{Decryption} $D'(c,PK')$: compute the plaintext
\begin{equation}
\label{eq:dec}
m = L(c \pmod{N^2}) \cdot \mu \pmod{N}.          	 
\end{equation}

\end{itemize}

The scheme described above resembles a digital signature scheme more than an encryption scheme, but we will 
follow the terminology from the previous works. Note that in the original Paillier scheme 
\cite{Paillier99}, the public key is $(N,g_p)$ and the secret key is $(\lambda,\mu)$. 
Hence, in the modified 
scheme described above $\mu$ is made a public parameter in order to make the decryption algorithm 
a) to be publicly computable (that is, make it like the verification algorithm of a digital signature scheme), and 
b) more efficient. The rationale to make $\mu$ public is the claim in \cite{NabeelABB13} that it is hard to
compute the discrete logarithm of $\mu$ w.r.t. base $g_p$ to obtain $\lambda$. 
Also, note that the modified scheme is also additively homomorphic, i.e., 
\[
E'(m_1 + m_2,r_1r_2) = E'(m_1,r_1) \cdot E'(m_2,r_2),
\]
and 
\[
{E'}^{m_2}(m_1,r_1) = E'(m_1 m_2,r_1^{m_2}).
\]

\subsection{Cryptanalysis}

We now show that the modified Paillier scheme is insecure. Namely, an adversary having access only to the 
public key can produce encryptions of messages, contrary to the claims in \cite{lpRideYu,NabeelABB13}
that only those who possess the secret key would be able to encrypt. Moreover, we also show that the 
ciphertexts are deterministic.   
  
\begin{lemma}
\label{lem:det}
The ciphertexts of the modified Paillier cryptosystem (see \eqref{eq:enc}) are deterministic. 
\end{lemma}
\begin{proof}
Since $c = g_p^{m\lambda}r^{N\lambda} \pmod{N^2}$, by the properties of the Carmichael function
$\lambda$,  we have $r^{N\lambda} \equiv 1 \pmod{N^2}$ (see for e.g., \cite[pp. 2]{Paillier99}). Hence $c = \left( g_p^{\lambda} \right)^m 
\pmod{N^2}$ and is independent of $r$. 
\qed
\end{proof}
 
Next, we show how to efficiently compute $g_p^{\lambda} \pmod{N^2}$ from the public key.
\begin{lemma}
\label{lem:gpl}
$g_p^{\lambda} \pmod{N^2} = N \cdot (\mu^{-1} \pmod{N}) + 1$. 
\end{lemma}
\begin{proof}
From \eqref{eq:mu}, 
\[
L(g_p^\lambda \pmod{N^2})\; \equiv\; \mu^{-1} \pmod{N}.
\]
Since $1 \le L(g_p^\lambda \pmod{N^2}) < N$ , the lemma follows from the definition of the $L$ function 
in \eqref{eq:l}. 
\qed
\end{proof}
From Lemmas \ref{lem:det} and \ref{lem:gpl} and their proofs, we have
\begin{corollary}
\label{cor:encm}
$E'(m) = (N \cdot (\mu^{-1} \pmod{N}) + 1)^m \pmod{N^2}$.
\end{corollary}

Hence, anyone with access to the public key can easily produce the (unique) ciphertext corresponding to a 
given message without knowing the secret key (or equivalently, will be able to forge the signature). Note 
that we were able to compute $g_p^{\lambda} \pmod{N^2}$ without explicitly computing $\lambda$. 
This constitutes a complete break of the modified Paillier cryptosystem. 

\section{Attack on the PP-CA-PSS from \cite{NabeelABB13}}
\label{sec:capss}

\subsection{Recall of the scheme}

For completeness, we briefly recall the relevant steps of the privacy-preserving
context-aware publish-subscribe system from \cite{NabeelABB13} that are necessary
to understand our attack. 
For other details, we refer the reader to 
\cite{NabeelABB13} and its full version \cite{fullNabeelABB13}. This protocol consists of the following types of entities: 
\begin{itemize}
\item\textbf{Context manager}: it is a trusted third party (TTP) responsible for initialising the system 
and registering other entities. There is a context manager for each context (e.g. location) and the 
context could change frequently. It provides secret keys to publishers (resp. subscribers) to encrypt 
(resp. decrypt) the content payload during the initialisation phase and there will be no further 
interaction with publishers/subscribers unless the system must be reinitialised. 
\item\textbf{Publisher}: owner of the messages/content that they like to publish and notify the subscribers.
\item\textbf{Subscriber}: entity interested to subscribe for the content from the publishers.
\item\textbf{Broker}: intermediary that matches the blinded/encrypted notifications and subscriptions, and
if there is a match, then forward the encrypted message to the corresponding subscriber.   
\end{itemize}

The threat model assumed in \cite{NabeelABB13} is that the context manager is fully trusted. 
The brokers are assumed to be honest-but-curious, i.e., they honestly follow the protocol but are curious to 
learn the confidential notifications and subscriptions. Publishers are expected to not collude
with any other entity and to follow the protocol honestly. Subscribers are not trusted. Brokers 
may collude with one another and also collude with malicious subscribers. 

The steps of the protocol are as follows:
\begin{itemize}
\item\textbf{System initialisation}: 
The context manager generates the parameters of the modified Paillier cryptosystem. It maintains
a set of contexts $\mathcal{C}$. Each context $C_i \in \mathcal{C}$ is a tuple
$C_i = ( \lambda_i, \mu_i, t_i, r_i )$, where $\lambda_i$ and $\mu_i$ are modified
Paillier parameters described in Section \ref{sec:modified}. Implicit in the description is the 
modulus $N_i$ that is different for each context, and $t_i,r_i \unif{\$}\Z{N_i}{}$ are random values.
Brokers only match notifications with subscriptions within the same context. The parameters $\lambda_i$
and $t_i$ are private to the context manager but $\mu_i$ (also $N_i$) are public.
  
The context manager also deploys an Attribute-based Group Key Management (AB-GKM) scheme \cite{NabeelB14} to manage
the secret keys issued to the subscribers that is used to decrypt the payload message of the notification. 
An AB-GKM scheme enables group key management while also enabling fine-grained access control among a 
group of users each of whom
is identified by a set of attributes. Subscribers need to provide their identity attributes (in an oblivious 
manner) to the context manager and obtain the secrets that is also shared with the subscribed publishers
by the context manager. 
These secrets are later used to 
derive the secret encryption key for the publisher to encrypt the payload content, and also derive decryption keys 
 for the subscribers. Since in this work we do not target the encrypted 
payload content, we omit the details corresponding to the AB-GKM scheme.
\vspace{0.25cm}
\item\textbf{Subscriber registration}: 
 every subscriber registers with the context manager. A subscriber with context $C_i$ receives the 
following parameters during the registration:
\[
(  E'(-r_i),\; E'(-1),\; g^{-t_i}E'(-r_i) ),
\]
where $g \unif{\$}\Z{N^2}{*}$. The subscriber uses these parameters to ``blind'' its subscriptions.
Since $\mu_i$ is public, the subscriber may decrypt $E'(-r_i)$ using
the decryption procedure $D'$ to obtain $r_i$. Note that the parameter $r_i$ is common to all the 
subscribers within a given context. It is claimed in  \cite{NabeelABB13} that the subscriber can neither 
recover $g^{-t_i}$ nor $t_i$ from $g^{-t_i}E'(-r_i)$. 
Needless to say, it was believed until now that the ciphertexts $E'(-r_i)$ are randomised. 

As briefly mentioned above, the subscribers also receive their secret of the AB-GKM scheme depending on the 
identity attributes they possess but without revealing them to the context manager. 
\vspace{0.25cm}
\item\textbf{Publisher registration}: 
every publisher too registers with the context manager. A publisher with context $C_i$ receives the 
following parameters during the registration:
\[
( E'(r_i),\; E'(1),\; g^{t_i}E'(r_i) ).
\] 
As in the case of subscribers, the publisher uses these parameters to blind its notifications.
In addition to the modified Paillier parameters, the publishers also receive the set of secrets for the 
AB-GKM scheme issued to the subscribers, from the context manager. The publisher uses these secrets to derive a secret encryption key
to selectively encrypt the payload of the notification of a subscriber depending on the latter's subscription. 
\vspace{0.25cm}
\item\textbf{Notifications}: 
in this system, every notification and subscription is represented as a Boolean expression over a set of 
attribute/value pairs. The publisher blinds a value $v$ $(0 \le v < 2^\ell \ll N_i) $ for an attribute $a$ 
as follows:
\begin{align}
\label{eq:not}
v' = &\; g^{t_i} \cdot E'(r_i) \cdot E'(r_i(v-1)) \cdot E'(r_v) \nonumber  \\
   =  &\; g^{t_i} \cdot E'(r_iv+r_v),
\end{align}
where $r_v$ is a suitably sampled random value in $\Z{N_i}{}$. For reasons that will be more clear 
when describing the matching phase, the value $r_v$ is chosen such that 
$ 0 \le r_i(v-x)+r_v \leq N_i/2$ for $x \in \{0,1,\ldots,v\}$, else $ N_i/2 < r_i(v-x)+r_v  < N_i $
for $x \in \{v+1,v+2,\ldots,2^\ell-1\}$.

The publisher also generates the encryption key $k$ of the 
AB-GKM scheme using the secret it received from the context manager for the corresponding subscriber.
The publisher then encrypts the payload (denoted as $\mathit{payload}$) of the notification 
$\langle a_i = v_i \rangle_{i \in I}$ as $\mathcal{E}_k(\mathit{payload})$.     
\vspace{0.25cm}
\item\textbf{Subscriptions}: 
suppose a subscriber wants to subscribe for the attribute $a$ with the value $x$. It blinds
$x$ as follows:
\begin{align}
\label{eq:sub}
x' = &\; g^{-t_i} \cdot E'(-r_i) \cdot E'(r_i(1-x)) \nonumber  \\
   =  &\; g^{-t_i} \cdot E'(-r_ix).
\end{align}
For each such attribute/value pairs $(a,x)$, a tuple $(a,x',\alpha)$, where $\alpha \in \{<,\geq\} $
is sent to the broker. Each such atomic subscription thus allows a homomorphic greater/less than 
comparison between the value and the notification. These atomic subscriptions are combined
using a Boolean formula to create the complete subscription.  
\vspace{0.25cm}
\item\textbf{Broker matching}: 
the brokers are assumed to know the public parameters for all the contexts $C_i$. Suppose a broker
receives the blinded subscription $x'$ and the blinded notification $v'$ for the attribute $a$.
It first computes $x'\cdot v'$ and, since the blinding values $g^{t_i}$ and $g^{-t_i}$ get cancelled, the product is now a typical
modified Paillier ciphertext that can be decrypted using the public parameters. It then decrypts
the ciphertext to obtain the randomised difference between the original values $x$ and $v$ as follows:
\begin{equation}
\label{eq:randdiff}
d' = D'(x'\cdot v') = r_i(v-x)+r_v.
\end{equation} 
The broker decides $v\geq x$ if $d' \le N_i/2$, otherwise $v<x$. For a full/composite 
subscription, the broker evaluates the Boolean formula. After successful matching, it forwards the 
encrypted payload $\mathcal{E}_k(\mathit{payload})$ to the subscriber. The subscriber having
valid credentials will be able to derive the secret key $k$ and decrypt the payload ciphertext.

\end{itemize}

\subsection{Attack on the scheme}

The basis of our attack on the PP-CA-PSS scheme from \cite{NabeelABB13} is our attack on the 
modified Paillier cryptosystem that we presented
in Section \ref{sec:modified}. 
First, we show that any registered publisher or subscriber will be able to compute the blinding value
$g^{t_i}$ used to hide notifications and subscriptions, respectively.  
\begin{lemma}
\label{lem:blind}
Within a context $C_i$, any publisher or subscriber will be able to efficiently compute $g^{t_i}$.  
\end{lemma} 
\begin{proof}
Consider the case of a subscriber who receives the tuple 
$(  E'(-r_i),\allowbreak\; E'(-1),\allowbreak\; g^{-t_i}E'(-r_i) )$
from the context manager during its registration.
From Lemma \ref{lem:det}, the value of $E'(-r_i)$ is  unique and so is the value $g^{-t_i} \cdot \allowbreak E'(-r_i)$.
Hence by dividing the latter by the former, the subscriber can easily recover the value $g^{-t_i}$.
Analogously, the publisher too can recover the blinding value.
\qed  
\end{proof}
Hence, the claim in  \cite{NabeelABB13} that publishers/subscriber cannot efficiently recover $g^{t_i}$ is incorrect.

\begin{corollary}
\label{cor:broker}
Within a context $C_i$, any broker, by colluding with a subscriber, will be able to efficiently recover 
the subscription $x$ from its blinded subscription $x'$. 
\end{corollary} 
The above corollary follows immediately from the above lemma since a broker can collude with any
subscriber and learn the blinding value $g^{t_i}$ (and hence $g^{-t_i}$) and the parameter $r_i$, and then remove the blinds from the blinded subscription $x'$ (see \eqref{eq:sub}) and eventually decrypt 
the ciphertext $E'(-r_ix)$ to recover $x$.  

The broker can also attempt to recover the notifications from publishers but since these values $v$ are 
blinded by the random values $r_v$ (see \eqref{eq:not}), we do not know how to recover the 
notifications. As already observed in \cite{NabeelABB13}, there will be a small leakage on the value
of $v-x$ leaked by the randomised difference $d'$ (see \eqref{eq:randdiff}) since $r_v$ is not sampled 
from the uniform distribution. Since we now know $x$, this directly translates to a small leakage on the 
value of the notification $v$ itself. Also, note that we do not get learn anything about the payload
of the notification.

\section{Attack on the lpRide RHS Protocol}
\label{sec:lpride}

\subsection{Recall of lpRide}

For completeness, we briefly recall the relevant steps of the lpRide RHS protocol that are necessary
to understand our attack. 
For other details, we refer to 
\cite{lpRideYu}. The lpRide protocol consists of the following entities: 
\begin{itemize}
\item\textbf{Authority}: responsible for the registration of riders and taxis/drivers.
\item\textbf{RSP}: the ride matching service provider that provides the online service of matching
riders' encrypted locations with the ``nearest'' taxis whose locations too are encrypted.
\item\textbf{Rider}: a rider $u$ provides the encryption of its location $l_u$ to the RSP to request a nearby taxi.
\item\textbf{Taxi} (\textit{aka.} driver): a taxi $t_k$ provides the encryption of its location $l_{t_k}$
to the RSP and waits for a match with its potential client.   
\end{itemize}

The threat model assumed in \cite{lpRideYu} is that the authority is fully trusted and that its 
communications with other entities happen over authenticated channels. The riders and taxis
are honest-but-curious
so is the  RSP. 
Further, it is assumed that
the RSP will not collude with any rider or taxi. 

The steps of the lpRide protocol are as follows:
\vspace{0.25cm}
\begin{itemize}
\item\textbf{System initialisation}: the RSP prepares the road network embedding \cite{ShahabiKS03} and 
partitions it into zones. The authority initialises the system parameters 
$(PK,SK) \unif{} \textrm{Init}_{\textrm{auth}}(\ell,\kappa)$, 
and by using the 
modified Paillier cryptosystem generates and broadcasts its public key
\begin{equation}
\label{eq:pklpr}
PK = (N,g_p,g,\bm{\mu},\ell,\kappa),
\end{equation}
where $\ell$ is the bit length of the location coordinates in the road network embedding, 
$\kappa-1 \le \log_2 N - 2 $ is the size of the random values used to blind the location coordinates of drivers, and 
$g \unif{\$} \Z{N^2}{*}$ is a random base, and $\bm{\mu} = \langle \mu_i \rangle_{0 \le i \le \omega}$
 is a $\omega$-tuple of modified Paillier cryptosystem public parameter $\mu$.
The secret key is 
\[
SK = (\bm{\lambda},\bm{\epsilon},\bm{\xi}),
\]
where $\bm{\lambda} = \langle \lambda_i \rangle_{0 \le i \le \omega}$
is a $\omega$-tuple of modified Paillier cryptosystem secret parameter $\lambda$. 
The parameters $\bm{\epsilon} = \langle \epsilon_i \rangle_{1 \le i \le \omega}$ and 
$\bm{\xi} = \langle \xi_i \rangle_{1 \le i \le \omega}$ are $\omega$-tuples of random values from 
$\Z{N}{}$.

\begin{remark}
\label{rem:n}
In the lpRide protocol, the same modulus $N$, and the same bases $g_p$ and $g$ are suggested 
to be used for all the coordinates $1 \le i \le \omega$. In this case, $\lambda_i$ are all identical, 
so are all $\mu_i$. But the protocol can be easily generalised to use different moduli and bases (and our attack applies for this variant too).
\end{remark}
\vspace{0.25cm}
\item\textbf{Rider and Taxi registration}:
\begin{itemize}
\vspace{0.15cm}
\item\textbf{Rider registration}: the authority assigns a rider $u$ its secret key 
$sk_u \unif{} \textrm{RiderKeyGen}(u,PK,SK)$. It chooses a set of random integers  
$\langle r_i \rangle_{0 \le i \le \omega}$ from $\Z{N}{*}$, and computes
\begin{align*}
E'(-1) & =  g_p^{-\lambda_0}r_0^{\lambda_0N} \pmod{N^2},\\
g^{-\epsilon_i}E'(-1) & =  g^{-\epsilon_i}g_p^{-\lambda_i}r_i^{\lambda_iN} \pmod{N^2},\\
g^{-\xi_i}E'(-1) & =  g^{-\xi_i}g_p^{-\lambda_i}r_i^{\lambda_iN} \pmod{N^2}.\\
\end{align*}
Then, the secret key of the rider is
\begin{equation}
\label{eq:sku}
sk_u =   (E'(-1), \langle g^{-\epsilon_i}E'(-1) \rangle_{1 \le i \le \omega}, 
         \langle g^{-\xi_i}E'(-1) \rangle_{1 \le i \le \omega} ).
\end{equation}
Note that $g^{-\epsilon_i}$ and $g^{-\xi_i}$ serve as blinding values. 
It is claimed in  \cite{lpRideYu}, 
as was done in \cite{NabeelABB13}, that the rider can neither recover $g^{-\epsilon_i}$ (resp. $g^{-\xi_i}$) 
nor ${\epsilon_i}$ (resp. ${\xi_i}$) from $g^{-\epsilon_i}E'(-1)$ (resp. $g^{-\xi_i}E'(-1)$). 
\vspace{0.15cm}
\item\textbf{Taxi registration}: the authority assigns a taxi $t_k$ its secret key 
$sk_{t_k} \unif{} \textrm{TaxiKeyGen}(t_k,PK,SK)$. It chooses a set of random integers  
$\langle r'_i \rangle_{0 \le i \le \omega}$ from $\Z{N}{*}$, and computes
\begin{align*}
E'(1) & =  g_p^{\lambda_0}{r'}_0^{\lambda_0N} \pmod{N^2},\\
g^{\epsilon_i}E'(1) & =  g^{\epsilon_i}g_p^{\lambda_i}{r'}_i^{\lambda_iN} \pmod{N^2},\\
g^{\xi_i}E'(1) & =  g^{\xi_i}g_p^{\lambda_i}{r'}_i^{\lambda_iN} \pmod{N^2}.\\
\end{align*}
Then, the secret key of the taxi is
\begin{equation}
\label{eq:skt}
sk_{t_k} = (E'(1), \langle g^{\epsilon_i}E'(1) \rangle_{1 \le i \le \omega}, 
        \langle g^{\xi_i}E'(1) \rangle_{1 \le i \le \omega} ).
\end{equation}
\end{itemize}
\item\textbf{Ride request}: 
a rider $u$ generates an encrypted ride request $\widehat{R}_u \unif{} \textrm{ReqGen}\allowbreak(\bm{c}_{u}, sk_u)$ by 
computing its 
location $\bm{c}_{u} = \langle \bm{c}_{u}[i] \rangle_{1 \le i \le \omega}$ in the embedded road network 
encoding, where $0 \le \bm{c}_{u}[i] < 2^\ell$ is an integer. Then, two vectors of \textit{blinded} 
ciphertexts $\widehat{\bm{c}}_u^+$ and $\widehat{\bm{c}}_u^-$ are computed from $\bm{c}_u$ and $sk_u$ by element-wise multiplication.
\begin{itemize}
\item $\widehat{\bm{c}}_u^+ = \langle \widehat{\bm{c}}^+_u[i] \rangle_{1 \le i \le \omega}$, where
\begin{equation}
\label{eq:chatrp}
\widehat{\bm{c}}_u^+[i] = sk_u[2] \cdot sk_u[1]^{-\bm{c}_u[i]-1} = g^{-\epsilon_i} E'(\bm{c}_u[i]).
\end{equation} 

\item $\widehat{\bm{c}}_u^- = \langle \widehat{\bm{c}}^-_u[i] \rangle_{1 \le i \le \omega}$, where
\begin{equation}
\label{eq:chatrm}
\widehat{\bm{c}}_u^-[i] = sk_u[3] \cdot sk_u[1]^{\bm{c}_u[i]-1} = g^{-\xi_i} E'(-\bm{c}_u[i]).
\end{equation} 
  
\end{itemize}
The rider sends $\widehat{R}_u = (\widehat{\bm{c}}_u^+, \widehat{\bm{c}}_u^-,z_u)$ to the RSP, where $z_u$ 
is the identity of the zone of the rider.
\vspace{0.25cm}
\item\textbf{Taxi location update}: 
a taxi $t_k$ computes its encrypted updated location $\widehat{L}_{t_k} \unif{} \textrm{TaxiUpdate}(\bm{c}_{t_k}, 
sk_{t_k})$ by computing its 
location $\bm{c}_{t_k} = \langle \bm{c}_{t_k}[i] \rangle_{1 \le i \le \omega}$ in the embedded road network 
encoding, where $0 \le \bm{c}_{t_k}[i] < 2^\ell$ is an integer. 
Two vectors $\bm{r}_1 = \langle \bm{r}_1[i]\rangle_{1 \le i \le \omega}$ and 
$\bm{r}_2 = \langle \bm{r}_2[i]\rangle_{1 \le i \le \omega}$
of $\kappa-1$-bit random integers are sampled. These random values are used to blind the driver locations and, eventually, the differences of the distances between rider and drivers. 
Then, two vectors of \textit{blinded} 
ciphertexts $\widehat{\bm{c}}_{t_k}^+$ and $\widehat{\bm{c}}_{t_k}^-$ are computed from $\bm{c}_{t_k}$ and $sk_{t_k}$ by 
element-wise multiplication.
\begin{itemize}
\item $\widehat{\bm{c}}_{t_k}^+ = \langle \widehat{\bm{c}}^+_{t_k}[i] \rangle_{1 \le i \le \omega}$, where
\begin{align}
\label{eq:chattp}
\widehat{\bm{c}}_{t_k}^+[i] & = sk_{t_k}[3] \cdot sk_{t_k}[1]^{\bm{c}_{t_k}[i]-1+\bm{r}_1[i]} \nonumber\\ 
& = g^{\xi_i} E'(\bm{c}_{t_k}[i]+\bm{r}_1[i]).
\end{align} 

\item $\widehat{\bm{c}}_{t_k}^- = \langle \widehat{\bm{c}}^-_{t_k}[i] \rangle_{1 \le i \le \omega}$, where
\begin{align}
\label{eq:chattm}
\widehat{\bm{c}}_{t_k}^-[i] & = sk_{t_k}[2] \cdot sk_{t_k}[1]^{-\bm{c}_{t_k}[i]-1+\bm{r}_2[i]} \nonumber\\ 
& = g^{\epsilon_i} E'(-\bm{c}_{t_k}[i]+\bm{r}_2[i]).
\end{align} 
  
\end{itemize}
The taxi $t_k$ sends $\widehat{L}_{t_k} = (\widehat{\bm{c}}_{t_k}^+, \widehat{\bm{c}}_{t_k}^-,z_{t_k})$ to the RSP, where $z_{t_k}$ is 
the identity of the zone of the taxi.
\vspace{0.25cm}
\item\textbf{Ride matching}:
the RSP receives the ride request $\widehat{R}_u$ and the set of taxi locations 
$\{\widehat{L}_{t_k}\}_{t_k \in \mathcal{T}}$ and filters the list of taxis based on the zone information of the rider.
Then, it computes the products $\widehat{\bm{c}}_{t_k}^+[i]\widehat{\bm{c}}_{u}^-[i]$ 
and $\widehat{\bm{c}}_{t_k}^-[i]\widehat{\bm{c}}_{u}^+[i]$ for all $1 \le i \le \omega$. Note that 
the blinding values $g^{\xi_i}$ or $g^{\epsilon_i}$ or their inverses that are present in the individual ciphertexts gets 
cancelled upon multiplication of the ciphertexts, and hence the products are typical modified Paillier ciphertexts that can be
decrypted with the public key $PK$. Note that the resulting difference of plaintexts is still blinded by the
random values $\bm{r}_1$ and $\bm{r}_2$. A detailed procedure to recover the sign of the differences from
these decrypted blinded sums is given \cite{lpRideYu} and we do not recall it here as it is not needed for 
the present purpose.   	

\end{itemize}

\subsection{Attack on lpRide}

The basis of our attack on lpRide is again our attack on the modified Paillier cryptosystem that we presented
in Section \ref{sec:modified}. In particular, Corollary \ref{cor:encm} shows that any one with the 
knowledge of the public key $PK$ of lpRide will be able to compute the modified Paillier ciphertexts.
First, we show that any registered rider or a taxi will be able to compute the secret keys of all the
other riders and drivers.
\begin{lemma}
\label{lem:skid}
The secret keys of riders (resp. taxis) are identical. Moreover, knowing any one secret key will suffice
to compute all the remaining secret keys of riders and drivers. 
\end{lemma} 
\begin{proof}
The first part of the lemma follows directly from Lemma \ref{lem:det}. For the second part, suppose
a rider $u$ has its secret key 
\begin{equation*}
sk_u =  (E'(-1), \langle g^{-\epsilon_i}E'(-1) \rangle_{1 \le i \le \omega}, 
       \langle g^{-\xi_i}E'(-1) \rangle_{1 \le i \le \omega} ).
\end{equation*}
From Corollary \ref{cor:encm}, it can efficiently compute the (unique value of) $E'(-1)$ and $E'(-1)$. 
Hence, any rider (analogously, any taxi) can easily determine all the values 
$g^{\epsilon_i}$, $g^{\xi_i}$. Once these values are obtained, it can compute the (identical) secret keys
of all the taxis. Similarly, any taxi can easily compute the secret keys of all the riders and other taxis. 
\qed  
\end{proof}
The above attack effectively makes the role of the trusted authority redundant as any registered rider or
taxi will be able to add others into the system without the consent of the authority.
\begin{lemma}
\label{lem:locr}
Any rider or taxi will be able to infer the locations of all the riders.
\end{lemma}
\begin{proof}
From the proof of Lemma \ref{lem:skid}, any rider or taxi can easily determine all the values 
$g^{\epsilon_i}$ and $g^{\xi_i}$. Using these computed values, it can unmask the ciphertexts 
(see \eqref{eq:chatrp} and \eqref{eq:chatrm})
that correspond to the locations of the riders, and eventually decrypt them using the public key $PK$. 
Note that these location values are not blinded by the random values 
$\bm{r}_1$ and $\bm{r}_2$ unlike the case of taxi location data. 
\qed
\end{proof}

\begin{remark}
Note that the RSP will not be able to learn the blinding values $g^{\epsilon_i}$ or $g^{\xi_i}$ or the 
exact locations of the riders since it assumed that it cannot collaborate with any rider or driver. 
However, if the RSP is allowed to disguise itself as a rider or a driver, then it will be able to learn
all the secret keys and locations of the rider.
\end{remark}

\section{Conclusion}
\label{sec:con}

We demonstrated an attack on the PP-CA-PSS from \cite{NabeelABB13} where any honest-but-curious 
entity (publisher/broker/subscriber) can learn the subscriptions of any subscriber.   
We also exhibited a key recovery attack on the lpRide RHS protocol \cite{lpRideYu}
by any  honest-but-curious rider or driver who can efficiently recover the secret keys of all other riders 
and drivers, and also learn the location of any rider. The basis of our attack is our 
cryptanalysis of the modified Paillier cryptosystem.
Since the notifications in the case of PP-CA-PSS protocol and the driver locations in the case
of lpRide are blinded by random values, we do not know how to recover these values, if at all 
it is possible. It will be interesting to explore this attack scenario.
Also, it will be interesting to explore
candidate constructions that are equivalent to the functionality of the modified Paillier cryptosystem but 
offer better security.

\section*{Acknowledgements}
This work was funded by the INSPIRE Faculty Award (by DST, Govt. of India) for the author.

\bibliography{morerefs}

\newcommand{\etalchar}[1]{$^{#1}$}
\begin{thebibliography}{NABB13b}

\bibitem[BLM{\etalchar{+}}19]{raza_bride}
M.~{Baza}, N.~{Lasla}, M.~{Mahmoud}, G.~{Srivastava}, and M.~{Abdallah}.
\newblock {B-R}ide: Ride sharing with privacy-preservation, trust and fair
  payment atop public blockchain.
\newblock {\em IEEE Transactions on Network Science and Engineering}, pages
  1--1, 2019.

\bibitem[CBDA{\etalchar{+}}]{cui2019}
Shujie Cui, Sana Belguith, Pramodya De~Alwis, Muhammad~Rizwan Asghar, and
  Giovanni Russello.
\newblock Collusion defender: Preserving subscribers' privacy in publish and
  subscribe systems.
\newblock {\em IEEE Transactions on Dependable and Secure Computing}.
\newblock To appear.

\bibitem[HNW{\etalchar{+}}18]{he_ridesharing}
Y.~{He}, J.~{Ni}, X.~{Wang}, B.~{Niu}, F.~{Li}, and X.~{Shen}.
\newblock Privacy-preserving partner selection for ride-sharing services.
\newblock {\em IEEE Transactions on Vehicular Technology}, 67(7):5994--6005,
  2018.

\bibitem[KFZC18]{khazbak}
Youssef Khazbak, Jingyao Fan, Sencun Zhu, and Guohong Cao.
\newblock {P}reserving {L}ocation {P}rivacy in {R}ide-{H}ailing {S}ervice.
\newblock In {\em 2018 {IEEE} Conference on Communications and Network
  Security, {CNS} 2018, Beijing, China, May 30 - June 1, 2018}, pages 1--9.
  {IEEE}, 2018.

\bibitem[KMV21]{DeepakMV19}
Deepak Kumaraswamy, Shyam Murthy, and Srinivas Vivek.
\newblock Revisiting driver anonymity in oride.
\newblock {\em CoRR}, abs/2101.06419, 2021.

\bibitem[LJFX19]{pRideLuo}
Yuchuan Luo, Xiaohua Jia, Shaojing Fu, and Ming Xu.
\newblock p{R}ide: {P}rivacy-{P}reserving {R}ide {M}atching {O}ver {R}oad
  {N}etworks for {O}nline {R}ide-{H}ailing {S}ervice.
\newblock {\em {IEEE} Trans. Information Forensics and Security},
  14(7):1791--1802, 2019.

\bibitem[Mun18]{munster2018}
Javier Munster.
\newblock Securing publish/subscribe.
\newblock Master's thesis, University of Toronto, 2018.
\newblock
  http://msrg.org/publications/pdf_files/2018/JMthesis-Securing_Publish-Subcribe.pdf.

\bibitem[NABB13a]{NabeelABB13}
Mohamed Nabeel, Stefan Appel, Elisa Bertino, and Alejandro~P. Buchmann.
\newblock Privacy preserving context aware publish subscribe systems.
\newblock In Javier L{\'{o}}pez, Xinyi Huang, and Ravi~S. Sandhu, editors, {\em
  Network and System Security - 7th International Conference, {NSS} 2013,
  Madrid, Spain, June 3-4, 2013. Proceedings}, volume 7873 of {\em Lecture
  Notes in Computer Science}, pages 465--478. Springer, 2013.

\bibitem[NABB13b]{fullNabeelABB13}
Mohamed Nabeel, Stefan Appel, Elisa Bertino, and Alejandro~P. Buchmann.
\newblock Privacy preserving context aware publish subscribe systems 2013-1.
\newblock Technical Report CCTECH-6, Cyber Center Technical Reports, Purdue
  University, 2013.

\bibitem[NB14]{NabeelB14}
Mohamed Nabeel and Elisa Bertino.
\newblock Attribute based group key management.
\newblock {\em Trans. Data Priv.}, 7(3):309--336, 2014.

\bibitem[Nor20]{norton_databreach}
NortonLifeLock.
\newblock {U}ber {A}nnounces {N}ew {D}ata {B}reach {A}ffecting 57 million
  {R}iders and {D}rivers.
\newblock
  \url{https://us.norton.com/internetsecurity-emerging-threats-uber-breach-57-million.html},
  2020.
\newblock Retrieved: April 10, 2020.

\bibitem[NSB12]{NabeelSB12}
Mohamed Nabeel, Ning Shang, and Elisa Bertino.
\newblock Efficient privacy preserving content based publish subscribe systems.
\newblock In Vijay Atluri, Jaideep Vaidya, Axel Kern, and Murat Kantarcioglu,
  editors, {\em 17th {ACM} Symposium on Access Control Models and Technologies,
  {SACMAT} '12, Newark, NJ, {USA} - June 20 - 22, 2012}, pages 133--144. {ACM},
  2012.

\bibitem[Pai99]{Paillier99}
Pascal Paillier.
\newblock Public-key cryptosystems based on composite degree residuosity
  classes.
\newblock In Jacques Stern, editor, {\em Advances in Cryptology - {EUROCRYPT}
  '99, International Conference on the Theory and Application of Cryptographic
  Techniques, Prague, Czech Republic, May 2-6, 1999, Proceeding}, volume 1592
  of {\em Lecture Notes in Computer Science}, pages 223--238. Springer, 1999.

\bibitem[PDE{\etalchar{+}}17]{ORidePaper}
Anh Pham, Italo Dacosta, Guillaume Endignoux, Juan~Ram{\'{o}}n
  Troncoso{-}Pastoriza, K{\'{e}}vin Huguenin, and Jean{-}Pierre Hubaux.
\newblock {OR}ide: {A} {P}rivacy-{P}reserving yet {A}ccountable
  {R}ide-{H}ailing {S}ervice.
\newblock In Engin Kirda and Thomas Ristenpart, editors, {\em 26th {USENIX}
  Security Symposium, {USENIX} Security 2017, Vancouver, BC, Canada, August
  16-18, 2017}, pages 1235--1252. {USENIX} Association, 2017.

\bibitem[PDJ{\etalchar{+}}17]{Pham2017PrivateRideAP}
Anh Pham, Italo Dacosta, Bastien Jacot{-}Guillarmod, K{\'{e}}vin Huguenin, Taha
  Hajar, Florian Tram{\`{e}}r, Virgil~D. Gligor, and Jean{-}Pierre Hubaux.
\newblock {P}rivate{R}ide: {A} {P}rivacy-{E}nhanced {R}ide-{H}ailing {S}ervice.
\newblock {\em PoPETs}, 2017(2):38--56, 2017.

\bibitem[SKS03]{ShahabiKS03}
Cyrus Shahabi, Mohammad~R. Kolahdouzan, and Mehdi Sharifzadeh.
\newblock A road network embedding technique for k-nearest neighbor search in
  moving object databases.
\newblock {\em GeoInformatica}, 7(3):255--273, 2003.

\bibitem[WZL{\etalchar{+}}18]{wangTrace}
F.~{Wang}, H.~{Zhu}, X.~{Liu}, R.~{Lu}, F.~{Li}, H.~{Li}, and S.~{Zhang}.
\newblock Efficient and privacy-preserving dynamic spatial query scheme for
  ride-hailing services.
\newblock {\em IEEE Transactions on Vehicular Technology}, 67(11):11084--11097,
  2018.

\bibitem[YJZ{\etalchar{+}}19]{yuPSRIde}
H.~{Yu}, X.~{Jia}, H.~{Zhang}, X.~{Yu}, and J.~{Shu}.
\newblock {PSR}ide: Privacy-preserving shared ride matching for online ride
  hailing systems.
\newblock {\em IEEE Transactions on Dependable and Secure Computing}, pages
  1--1, 2019.

\bibitem[YSJ{\etalchar{+}}19]{lpRideYu}
Haining Yu, Jiangang Shu, Xiaohua Jia, Hongli Zhang, and Xiangzhan Yu.
\newblock lp{R}ide: Lightweight and privacy-preserving ride matching over road
  networks in online ride hailing systems.
\newblock {\em {IEEE} Trans. Vehicular Technology}, 68(11):10418--10428, 2019.

\end{thebibliography}
\bibliographystyle{alpha}

\hyphenation{op-tical net-works semi-conduc-tor}

%
%
\end{document}